\documentclass[%
 reprint,
 amsmath,amssymb,
 aps,
]{revtex4-2}

\usepackage{pkgfile}
\usepackage{xifthen}
\usepackage{tabularx}
\usepackage{adjustbox}
\usepackage{bbm}
\usepackage{makecell}

\usepackage{tikz}
\usetikzlibrary{matrix,arrows}

\usepackage{graphicx}
\usepackage{dcolumn}
\usepackage{bm}

\begin{document}


	\title[Hamiltonian Criterion] {Determination of the Hamiltonian from the Equations of Motion with Illustration from Examples}

\author{Chung-Ru Lee}
\affiliation{Department of Mathematics\\
National University of Singapore.}
\email{math.crlee@gmail.com}

\date{\today}
\counterwithout{equation}{section}

\begin{abstract}
In this paper, we study the determination of Hamiltonian from a given equations of motion. It can be cast into a problem of matrix factorization after reinterpretation of the system as first-order evolutionary equations in the phase space coordinates. We state the criterion on the evolution matrix for a Hamiltonian to exist. In addition, the proof is constructive and an explicit Hamiltonian with accompanied symplectic structure can be obtained. As an application, we will study a few classes of dynamical systems for illustration.
\end{abstract}
\date{\today}
\maketitle


\section{Introduction}
It is known that a damped oscillator
\begin{align}
\ddot{x}+\gamma\dot{x}+x=0\qquad(\gamma>0)\label{eq:damped}
\end{align}
appears to be a dissipative system and cannot be described as a Hamiltonian system. A common explanation for this is to observe that the total energy $\frac{1}{2}\dot{x}^2+\frac{1}{2}x^2$ is not conserved under time change.

Clearly the total energy (that is, the kinetic energy plus the potential) is non-conserved, as
$$\frac{d}{dt}(\frac{1}{2}\dot{x}^2+\frac{1}{2}x^2)=-\gamma\dot{x}^2<0$$
Decreases over time. However, it is still reasonable to suspect the existence of some other conserved quantities that could be treated as the Hamiltonian.

A formal way to prove that (\ref{eq:damped}) is non-Hamiltonian is to employ the symplectic framework \cite{arnold,da_silva,woodhouse}. Introduce the Hamiltonian vector field
\begin{align}\mathbf{X}_H=\frac{d}{dt}=\dot{p}\frac{\partial}{\partial p}+\dot{x}\frac{\partial}{\partial x}.\end{align}
Then it suffices to show that there does not exist an $H=H(\mathbf{p},\mathbf{x})$ so that the fundamental equation of Hamiltonian mechanics
\begin{align}
i_{\mathbf{X}_H}\omega=-dH\label{eq:fundamental}
\end{align}
holds, where $\omega$ is the canonical symplectic $2$-form $\omega=dp\wedge dx$.

Indeed, from equation (\ref{eq:damped}), we may put $\dot{x}=p$ and $\dot{p}=-\gamma p-x$. Thus, $i_{\mathbf{X}_H}\omega=-pdp-xdx-\gamma pdx$, and it is not a closed $1$-form.

Nevertheless, Bateman \cite{bateman} made the intelligent observation that a system of coupled oscillators with balanced energy loss-gain can be shown to be Hamiltonian (with a non-standard Hamiltonian form). The coupled equations of motion for the described system is:
\begin{align}
\begin{cases}
\ddot{x}+\gamma\dot{x}+x=-\lambda y\\
\ddot{y}-\gamma\dot{y}+y=-\lambda x.
\end{cases}\label{eq:bender1}
\end{align}
The parameter $\lambda$ here indicates the coupling strength.

Bender \textit{et al.} \cite{bender1,bender2} studied this system further and concluded that these equations can be derived using a non-standard quadratic Hamiltonian
$$H=pq+\frac{1}{2}\gamma(yq-xp)+(1-\frac{1}{4}\gamma^2)xy+\frac{1}{2}\lambda(x^2+y^2).$$
Note that in convention $p=\dot{x}$ is the momentum in $x$-direction and $q=\dot{y}$ is the momentum in $y$-direction. However, here one would observe that this is not the case. In fact, $q=\dot{x}+\frac{\gamma}{2}x$ and $p=\dot{y}-\frac{\gamma}{2}y$ when the canonical Hamiltonian equations are applied.

For convenience, we will also use $\mathbf{x}=(x_1,x_2)$ and $\mathbf{p}=(p_1,p_2)$ to denote $(x,y)$ and $(p,q)$ respectively. We will also write $\xi=(\mathbf{p},\mathbf{x})$ to denote the phase space coordinates.

Bender \textit{et al.} \cite{bender3} also noticed that from the system of equations (\ref{eq:bender1}),
$$pq+xy+\frac{\lambda}{2}(x^2+y^2)$$
is a conserved quantity under the conventional choice of momentum $\mathbf{p}=\dot{\mathbf{x}}$.

A subtle point here is that if one treats the above conserved quantity as the Hamiltonian of the system, one would not reach the correct equations of motion using the canonical Hamilton formulation.

This leads us to consider the question below: does there exist a suitably chosen symplectic structure so that (\ref{eq:fundamental}) with the Hamiltonian defined as $H:=pq+xy+\frac{\lambda}{2}(x^2+y^2)$ can produce the equations of motion
\begin{align}
\begin{cases}
\dot{p}=-\gamma p-x-\lambda y\\
\dot{q}=\gamma q-\lambda x-y\\
\dot{x}=p\\
\dot{y}=q
\end{cases}\label{eq:ad_hoc}
\end{align}
which is equivalent to (\ref{eq:bender1}). The answer to this question is positive: simply let
\begin{align}
\omega=dp\wedge dy+dq\wedge dx+\gamma dx\wedge dy,
\end{align}
from which the (nonzero) Poisson bracket follows:
$$\{p,q\}=-\gamma\qquad\{p,y\}=-1\qquad\{q,x\}=-1.$$

Another convenient way to state equation (\ref{eq:fundamental}) is
\begin{align}\label{eq:hamiltonian}
\frac{d}{dt}\xi_i=\{\xi_i,H\},
\end{align}
Here $\xi_i$ can be any entry of $\xi=\begin{pmatrix}
\mathbf{p}\\
\mathbf{x}
\end{pmatrix}$. Equation (\ref{eq:ad_hoc}) is then a natural consequence of (\ref{eq:hamiltonian}) using the Poisson brackets above.

We can see, a dynamical system has two pieces, namely the symplectic form and the Hamiltonian. As emphasized by Souriau, to reestablish the equations of motion, what we should do is by varying either (or both) parts of the dynamical system \cite{carinena}.

We would like to study the extent of this proposal. The purpose of this paper is two-fold. We would like to provide a criterion for identifying Hamiltonian systems when given some equations of motion. We also want our proposition to be constructive: that is, coming up with a recipe for generating a corresponding set $(H,\omega)$ of quadratic Hamiltonian and symplectic structure for such a system.

\section{Evolution Matrix and the Criterion}
From the description above, three immediate questions emerge:
\begin{enumerate}
\item How to decide whether a system is Hamiltonian, given its equations of motion?
\item Is there a recipe for construction of the pair $(H,\omega)$ for such a given system?
\item When the association to $(H,\omega)$ is non-unique, what is the relation between the different sets? In particular, does there exist canonical representatives?
\end{enumerate}
We shall provide an affirmative answer to all of these questions above.

For most physically interesting models, the Hamiltonian would be of the form
\begin{align}
H=H_0+V,
\end{align}
where $H_0$ is the quadratic part, and $V=V(\mathbf{x})$ is a function of only the position (that comes from the potential energy). Then, the part of $V$ would not jeopardize the symplectic structure and one can adjust the Hamiltonian to include the non-linear interactions in the Hamiltonian flow.

In light of that, we will consider only the quadratic part of the Hamiltonian and start with the assumption that $H=H_0$. Equivalently, the equations of motion we consider will be linear.

For a system of (linear second-order) equations of motion
\begin{align}
\ddot{\mathbf{x}}-B_1\dot{\mathbf{x}}-B_2\mathbf{x}=0\label{eq:motion}
\end{align}
we can always (but non-uniquely) reduce it to some first-order evolutionary equations. If we set the phase space coordinates to be $\xi=\begin{pmatrix}
\mathbf{p}\\ \mathbf{x}
\end{pmatrix}$, then this is to say that there exists some $M\in\mathrm{M}_{2n}(\mathbb{R})$ so that
$$\dot{\xi}=M\xi$$
is equivalent to (\ref{eq:motion}).

By a direct computation
$$M=2\underline{\omega}^{-1}\underline{H},$$
where $\underline{\omega}=\left(\omega(\frac{\partial}{\partial\xi_i},\frac{\partial}{\partial\xi_j})\right)_{ij}$ and $\underline{H}=\frac{1}{2}\left(\frac{\partial^2}{\partial\xi_j\partial\xi_i}H\right)_{ij}$. Those matrices constitute the symplectic structure and the Hamiltonian, respectively. Therefore, a system is Hamiltonian only when one can write $M$ as a product of alternating and symmetric matrices.

\begin{defn}
We say that a system with evolutionary equations $\dot{\xi}=M\xi$ (referred to as the system of $M$ in short) is \textit{admissible} if both $M$ and $M_{21}$ are invertible, where $M=\begin{pmatrix}
M_{11} & M_{12}\\
M_{21} & M_{22}
\end{pmatrix}$.
\end{defn}

The Main Theorems of this paper consists of three parts:
\begin{thm}\label{thm:main}
A system of first-order evolutionary equations
$$\dot{\xi}=M\xi$$
allows a Hamiltonian $H$ with accompanied symplectic structure $\omega$ if and only if $M\in\mathrm{GL}_{2n}(\mathbb{R})$ satisfies the criterion that $M$ and $-M$ are conjugate.
\end{thm}

\begin{thm}\label{thm:form}
For any system $M=\begin{pmatrix}
M_{11} & M_{12}\\
M_{21} & M_{22}
\end{pmatrix}$ with $M_{21}\in\mathrm{GL}_n(\mathbb{R})$, one can find another system
\begin{align}
\dot{\xi}=M_\mathrm{std}\cdot\xi
\end{align}
so that the associated equations of motion are the same, and $$M_\mathrm{std}=\begin{pmatrix}
B_1 & B_2\\
\mathbbm{1}_n & \mathbf{0}_n
\end{pmatrix}$$
is conjugate to $M$.
\end{thm}

\begin{thm}\label{thm:surgery}
Suppose
\begin{align}
M\xi=\begin{pmatrix}
M_{11} & M_{12}\\
M_{21} & M_{22}
\end{pmatrix}\begin{pmatrix}
\mathbf{p}\\ \mathbf{x}
\end{pmatrix},
\end{align}
and $M_{21}$ is invertible. If
\begin{align*}
B_1&=M_{21}M_{11}M_{21}^{-1}+M_{22}\\B_2&=M_{21}M_{12}-M_{21}M_{11}M_{21}^{-1}M_{22}.\end{align*}
satisfies
\begin{enumerate}
\item[$1.$] $B_1$ allows a decomposition $B_1=A_1S_1$ into a product of alternating and symmetric matrices and that
\item[$2.$] $S_1\in\mathrm{GL}_n(\mathbb{R})$ and $S_1B_2$ is symmetric.
\end{enumerate}
then the second-order equations of motion derived from the evolutionary equations allows a Hamiltonian $H_\mathrm{can}$ that is equipped with the canonical symplectic structure $\omega=\sum_{i}dp_i\wedge dx_i$.
\end{thm}

We provide detailed proofs and calculations in the next section. To recap: Theorem \ref{thm:main} states the criterion for which a system allows a Hamiltonian structure \cite{rodman}. The other two Theorems concern the problem of shaping the corresponding $(H,\omega)$ at will. Theorem \ref{thm:form} allow us to retain the conventional relation for the momentum $\mathbf{p}=\dot{\mathbf{x}}$ of the evolution matrix for a given admissible system. Theorem \ref{thm:surgery} is about a special but common case where one can choose the canonical symplectic structure.

From Theorem \ref{thm:main} we know that a system is Hamiltonian only if the characteristic polynomial is an even function. That is,
\begin{align}
\det(t\mathbbm{1}_{2n}-M)=\sum_{m=0}^{n}a_{2m}t^{2m}.
\end{align}
As a quick verification, we consider again the damped oscillator (\ref{eq:damped}). In this case
$$M=\begin{pmatrix}
-\gamma & -1\\
1 & 0
\end{pmatrix}$$
and so the characteristic polynomial is $t^2+\gamma t+1$, which can be Hamiltonian only when $\gamma=0$.

\section{Proofs of the Theorems}
We prove the three Theorems stated above.
\begin{proof}[Proof of Theorem \ref{thm:main}]
Suppose $M=AS$, where $A$ is alternating and $S$ is symmetric. Then any conjugate of $M$ would allow such decomposition. In fact, the decomposition $M=AS$ is compatible with $\mathrm{GL}_{2n}(\mathbb{R})$-conjugation in the sense that
$$\Lambda^{-1}M\Lambda=(\Lambda^{-1}A\Lambda^{-t})(\Lambda^{t}S\Lambda),$$
where $\Lambda^{-1}A\Lambda^{-t}$ (resp. $\Lambda^{t}S\Lambda$) is still antisymmetric (resp. symmetric).

Therefore, it suffices to construct the decomposition for one representative in the conjugacy class. In \cite{rodman}, such decomposition can be found for the rational canonical form of $M$. Let us describe the procedure here.

Suppose $M\in\mathrm{GL}_{2n}(\mathbb{R})$ is in the rational canonical form, namely it is by assumption a direct sum $M=\bigoplus_{k}M_k$ of matrices of the form
$$M_k^t=\begin{pmatrix}
0 & 1 & 0 & 0 & 0 & \dots & 0 & 0\\
0 & 0 & 1 & 0 & 0 & \dots & 0 & 0\\
\vdots & \vdots & \vdots & \vdots & \vdots & & \vdots & \vdots\\
0 & 0 & 0 & 0 & 0 & \dots & 0 & 1\\
a_0 & 0 & a_2 & 0 & a_4 & \dots & a_{2r_k-2} & 0
\end{pmatrix}$$
with $r_k\in\mathbb{Z}$. Note that we wrote the transpose of $M_k$ just to save space.

In this case $M_k$ has a decomposition $M_k=A_kS_k$, where
\begin{widetext}
\begin{align*}
A_k & =\begin{pmatrix}
0 & -1\\
1 & 0
\end{pmatrix}
\oplus\begin{pmatrix}
0 & -y_{r_k-2} & 0 & -y_{r_k-3} & \dots & -y_1 & 0 & -y_0\\
y_{r_k-2} & 0 & y_{r_k-3} & \reflectbox{$\ddots$} & \reflectbox{$\ddots$} & 0 & y_0 & \\
0 & -y_{r_k-3} & \reflectbox{$\ddots$} & \reflectbox{$\ddots$} & \reflectbox{$\ddots$} & -y_0 & & \\
\vdots & \reflectbox{$\ddots$} & \reflectbox{$\ddots$} & \reflectbox{$\ddots$} & \reflectbox{$\ddots$} & & & \\
y_0 & & & & & & &
\end{pmatrix},\\
S_k & =\begin{pmatrix}
1
\end{pmatrix}\oplus\begin{pmatrix}
& & & & & & & & -b_0\\
& & & & & & & b_0 & 0\\
& & & & & & -b_0 & 0 & b_1\\
& & & & & \reflectbox{$\ddots$} & 0 & -b_1 & 0\\
& & & & \reflectbox{$\ddots$} & 0 & b_1 & 0 & b_2\\
& & & \reflectbox{$\ddots$} & \reflectbox{$\ddots$} & \reflectbox{$\ddots$} & \reflectbox{$\ddots$} & \reflectbox{$\ddots$} & \vdots\\
-b_0 & 0 & b_1 & 0 & b_2 & 0 & \dots & \dots & b_{r_k-1}
\end{pmatrix}.
\end{align*}
\end{widetext}
Here $b_i$ and $y_i$ satisfies
\begin{align*}
b_0 & =a_0,\\
y_0 b_0 & = 1,
\end{align*}
in addition to
\begin{align*}
b_0y_m-\sum_{l=1}^mb_{l}y_{m-l} & =0,\\
\sum_{l=1}^mb_{l}y_{m-l} & =-a_{2r_k-2-2m}.
\end{align*}
Here the domain of $m$ is $m=1,2,\dots,r_k-2$ for the first equation, and an extra term of $m=r_k-1$ for the second (so there are a total of $2r_k-1$ variables and $2r_k-1$ equations). Note that from these equations the matrices $A_k$ and $S_k$ can be solved.
\end{proof}

\begin{proof}[Proof of Theorem \ref{thm:form}]
For any admissible
$$M=\begin{pmatrix}
M_{11} & M_{12}\\
M_{21} & M_{22}
\end{pmatrix},$$
adjust so that $M_{21}=\mathbbm{1}_n$ with a conjugation by $\begin{pmatrix}
M_{12} & \\ & \mathbbm{1}_n
\end{pmatrix}$. Then, $M_{22}=\mathbf{0}$ can be achieved through a $\begin{pmatrix}
\mathbbm{1}_n & M_{22}\\ & \mathbbm{1}_n
\end{pmatrix}$-conjugation.
\end{proof}

Before we prove Theorem \ref{thm:surgery}, let us state and proof a lemma:
\begin{lem}\label{lem:surgery}
Two admissible evolution systems $M$ and $W$ are in association to the same equations of motion if and only if they are equivalent under a $P$-conjugation, where $P$ is a subgroup of $\mathrm{GL}_{2n}(\mathbb{R})$ of the form
$$P=\left\{\begin{pmatrix}
T & X\\
& \mathbbm{1}_n
\end{pmatrix}\middle\vert T\in\mathrm{GL}_n(\mathbb{R}),X\in\mathrm{M}_n(\mathbb{R})\right\}.$$
\end{lem}
\begin{proof}
By a computation that is straightforward, we know that the system $\dot{\xi}=M\xi$ is associated to
\begin{align*}
\ddot{\mathbf{x}}-B_1\dot{\mathbf{x}}-B_2\mathbf{x}=0,
\end{align*}
where $B_1$ and $B_2$ are defined as in the Theorem statement. Moreover, $B_1$ and $B_2$ are invariant under $P$-conjugation.

This, along with Theorem \ref{thm:form}, proves that if $M$ and some $W$ are associated to the same set of equations of motion (\ref{eq:motion}), then they are both $P$-conjugate to
$$\begin{pmatrix}
B_1 & B_2\\
\mathbbm{1}_n & \mathbf{0}_n
\end{pmatrix}$$
and are therefore $P$-conjugates themselves.
\end{proof}
\begin{remark}\label{rmk:quotient}
Let $\mathcal{A}\subset\mathrm{GL}_{2n}(\mathbb{R})$ be the set of admissible matrices of dimension $2n\times 2n$. In simple words, Lemma \ref{lem:surgery} says that the map
\begin{align*}
\varphi:\mathcal{A} & \rightarrow\mathbb{R}^{n\times2n}\\
M & \mapsto(B_1,B_2)
\end{align*}
factors through the quotient map
$$\pi:\mathcal{A}\rightarrow\mathcal{A}/P.$$
Moreover, the pushforward map $\mathcal{A}/P\rightarrow\mathbb{R}^{n\times2n}$ is a bijection since the stabilizer of any such $M_\mathrm{std}$ in Theorem \ref{thm:form} intersects $P$ trivially.
\end{remark}
From the perspective of Remark \ref{rmk:quotient}, it should be natural to ask if there exists a section for the map $\varphi$. It is noted that Theorem \ref{thm:form} answered the question for us.

\begin{proof}[Proof of Theorem \ref{thm:surgery}]
By Theorem \ref{thm:form}, we may assume that $M$ is in the form
$$M=\begin{pmatrix}
B_1 & B_2\\
\mathbbm{1}_n & \mathbf{0}_n
\end{pmatrix}.$$

Assume the conditions mentioned in the Theorem statement. Then with $S_2:=-S_1B_2$, we have
\begin{align}
M=\begin{pmatrix}
B_1 & B_2\\ \mathbbm{1}_n & \mathbf{0}
\end{pmatrix}=\begin{pmatrix}
A_1 & -S_1^{-1}\\ S_1^{-1} & \mathbf{0}
\end{pmatrix}\begin{pmatrix}
S_1 & \\ & S_2
\end{pmatrix}.\label{eq:semi-canonical}
\end{align}
Furthermore, we can decompose additively $S_1A_1S_1=S_1B_1=X-X^t$ for some $X\in M_{n}(\mathbb{R})$ (for instance, one can let $X=\frac{1}{2}S_1B_1$). It follows that $M$ is $P$-conjugate to
$$\begin{pmatrix}
 & -\mathbbm{1}_n\\ \mathbbm{1}_n &
\end{pmatrix}\begin{pmatrix}
S_1^{-1} & S_1^{-1}X \\ X^tS_1^{-1} & X^tS_1^{-1}X+S_2
\end{pmatrix}=:M_\mathrm{can}.$$
The second factor in the left-hand side product defines $H_\mathrm{can}$.
\end{proof}

\section{A few illustrative examples}
\subsection{A dual to the system in (\ref{eq:bender1})}
We consider a dynamical system
\begin{align}
\begin{cases}
\ddot{x}+\gamma\dot{y}+x=-\lambda y\\
\ddot{y}-\gamma\dot{x}+y=-\lambda x,
\end{cases}\label{eq:new}
\end{align}
which can be written into the evolutionary equations with $M$ being
$$M=\begin{pmatrix}
0 & -\gamma & -1 & -\lambda\\
\gamma & 0 & -\lambda & -1\\
1 & 0 & 0 & 0\\
0 & 1 & 0 & 0
\end{pmatrix}.$$
Note that is satisfies both the criterion of Theorem \ref{thm:main}.
We yield the decomposition
$$M=\begin{pmatrix}
0 & -\gamma & -1 & 0\\
\gamma & 0 & 0 & -1\\
1 & 0 & 0 & 0\\
0 & 1 & 0 & 0
\end{pmatrix}\begin{pmatrix}
1 & 0 & 0 & 0\\
0 & 1 & 0 & 0\\
0 & 0 & 1 & \lambda\\
0 & 0 & \lambda & 1
\end{pmatrix}.$$
This is equivalent to saying that the system allows a Hamiltonian
$$H=\frac{1}{2}(p^2+q^2)+\frac{1}{2}(x^2+y^2)+\lambda xy$$
under the symplectic structure given by
$$\omega=dp\wedge dx+dq\wedge dy-\gamma dx\wedge dy.$$

By Theorem \ref{thm:surgery}, observe that the system (\ref{eq:new}) has a Hamiltonian $H_\mathrm{can}$:
$$\frac{1}{2}(p^2+q^2)+\frac{\gamma}{2}(qx-py)+\lambda xy+\frac{1}{2}(1+\frac{\gamma^2}{4})(x^2+y^2)$$
with the canonical symplectic $2$-form on the phase space (that derives the canonical Hamiltonian equations).

To apply the standard stability analysis of the evolution matrix, we find the characteristic polynomial of $M$:
$$\det(t\mathbbm{1}_{4}-M)=t^4+(2+\gamma^2)t^2+(1-\lambda^2),$$
which has zeros at
$$t^2=-\frac{1}{2}(2+\gamma^2\pm\sqrt{\gamma^4+4\gamma^2+4\lambda^2}).$$
Oscillatory behavior occurs when $t^2<0$, which happens when $\lambda<1$. In contrast to Bender's case, the critical value for the existence of PT-symmetric phase in this model is independent of $\gamma$.

\subsection{An interactive pair of PT-symmetric systems}
Consider a system with two cross-coupled PT-symmetric pairs of (\ref{eq:bender1}), equipped with possibly distinct energy loss-gain coefficient. Here, the parity alternation means to switch $x_i$ and $y_i$, for both $i=1,2$:
\begin{align}
\begin{cases}
\ddot{x}_1+\gamma_1\dot{x}_1+x_1=-\lambda_1 x_2-\lambda_2 y_2\\
\ddot{y}_1-\gamma_1\dot{y}_1+y_1=-\lambda_2 x_2-\lambda_1 y_2\\
\ddot{x}_2+\gamma_2\dot{x}_2+x_2=-\lambda_1 x_1-\lambda_2 y_1\\
\ddot{y}_2-\gamma_2\dot{y}_2+y_2=-\lambda_2 x_1-\lambda_1 y_1,
\end{cases}\label{eq:interaction}
\end{align}
It also satisfies the criterion of Theorem \ref{thm:main} and we may check that (\ref{eq:interaction}) comes from the pair $(H,\omega)$ with
\begin{align*}
H =&\sum_{i=1}^2(p_iq_i+x_iy_i)+\\
&\lambda_1(x_1y_2+x_2y_1)+\lambda_2(x_1x_2+y_1y_2)\\
\omega =&\sum_{i=1}^2(dp_i\wedge dy_i+dq_i\wedge dx_i)+\\
&\gamma_1dx_1\wedge dy_1+\gamma_2dx_2\wedge dy_2.
\end{align*}
Moreover, we can also write down the Hamiltonian $H_\mathrm{can}$ that is accompanied by the canonical symplectic $2$-form:
\begin{align*}
&\sum_{i=1}^2\Big(p_iq_i+(1-\frac{1}{4}\gamma_i^2)(x_iy_i)-\frac{1}{2}\gamma_i(p_ix_i-q_iy_i)\Big)+\\
&\lambda_1(x_1y_2+y_1x_2)+\lambda_2(x_1x_2+y_1y_2).
\end{align*}
The general discussion toward the characteristics of this system is a complex task. However, one can observe qualitatively that the oscillatory solutions occur only when the damping is light relative to the coupling constants. Moreover, when the coupling is weak, each subsystem $(x_1,y_1)$ and $(x_2,y_2)$ should behave almost independently. In this case the behaviour depends mainly on $\gamma_1$ and $\gamma_2$.

\subsection{Higher degree terms: the H\'{e}non-Heiles system}
As an example for dealing with cases with higher degree terms, consider the H\'{e}non-Heiles system with balanced energy loss-gain
\begin{align}
\begin{cases}
\ddot{x}+\gamma\dot{y}+x=x^2-y^2\\
\ddot{y}-\gamma\dot{x}+y=-2xy,
\end{cases}
\end{align}
with the recognition from (\ref{eq:new}), we try to write it as an Hamiltonian system with
$$H=\frac{1}{2}(p^2+q^2)+\frac{1}{2}(x^2+y^2)+V(x,y).$$
along with the symplectic structure defined by
$$\omega=dp\wedge dx+dq\wedge dy-\gamma dx\wedge dy.$$
Now, re-expressing in evolutionary equations,
\begin{align}
\begin{cases}
\dot{p}=-\gamma q-x+x^2-y^2\\
\dot{q}=\gamma p-y-2xy\\
\dot{x}=p\\
\dot{y}=q
\end{cases}\label{eq:higher}
\end{align}
We can calculate the potential term by the fundamental equation (\ref{eq:hamiltonian}). It satisfies
$$\begin{cases}
-\frac{\partial}{\partial x}V=x^2-y^2\\
-\frac{\partial}{\partial y}V=-2xy
\end{cases}$$
Solving to get $V=y^2x-\frac{1}{3}x^3$ is then straightforward.

One can tune the parameter $\gamma$ to control the chaotic behavior of the system.

\section{A Remark: Relation to Lagrangian Mechanics}
Despite its original root in Hamiltonian mechanics, the work in this paper can be adopted in the setting of Lagrangian mechanics to provide a straightforward algorithm for the construction of Lagrangian from the equations of motion once the pair $(H,\omega)$ is found. In fact, the Lagrangian $L$ can be yielded from
\begin{align}
L=\theta(\mathbf{X}_H)-H,\label{eq:lagrangian}
\end{align}
where $\theta$ is a $1$-form satisfying $d\theta=\omega$. Note the formula we apply above is different from the usual Legendre transformation. We write it so to incorporate the possible case of non-canonical symplectic $2$-forms.

To elaborate, let us consider a second-order differential equation
\begin{align}
\ddot{\mathbf{x}}=B_1{\dot{\mathbf{x}}}+B_2\mathbf{x},\label{eq:our_equation}
\end{align}
where $B_1$ and $B_2$ are constant matrices satisfying the conditions of Theorem \ref{thm:surgery}. Note that the system allows trivially a standard $M_\text{std}$ and therefore $\mathbf{p}=\dot{\mathbf{x}}$ under such basis. If we apply the decomposition from (\ref{eq:semi-canonical}) to the equation (\ref{eq:lagrangian}), we have
$$\underline{\omega}^{-1}=\begin{pmatrix}
\mathbf{0} & S_1\\
-S_1 & S_1A_1S_1
\end{pmatrix}$$
and therefore
$$\theta(\mathbf{X}_H)=\dot{\mathbf{x}}^t S_1\dot{\mathbf{x}}+\mathbf{x}^t S_1A_1S_1\dot{\mathbf{x}}.$$
On the other hand,
$$H=\frac{1}{2}(\dot{\mathbf{x}}^t S_1\dot{\mathbf{x}}-\mathbf{x}^t S_2\mathbf{x})$$
is straightforward. To conclude, subtract and see
$$L=\frac{1}{2}\left(\dot{\mathbf{x}}^tS_1\dot{\mathbf{x}}+\mathbf{x}^tS_1B_1\dot{\mathbf{x}}+\mathbf{x}^tS_1B_2\mathbf{x}\right).$$

Of course, the choice of the Lagrangian $L$ here would be up to any a full time-derivative. As an example, for our model system (\ref{eq:new}), the Lagrangian can be chosen to be
\begin{align*}
L & =\frac{1}{2}(\dot{x}^2+\dot{y}^2)-\frac{1}{2}\gamma(x\dot{y}+y\dot{x})-\frac{1}{2}(x^2+y^2+2\lambda xy)
\end{align*}
or equivalently,
\begin{align*}
L&=\frac{1}{2}(\dot{x}^2+\dot{y}^2)-\frac{1}{2}(x^2+y^2)-\gamma x\dot{y}-\lambda xy
\end{align*}
by the formula (\ref{eq:lagrangian}).

\section*{Acknowledgments}
The author wishes to express sincere gratitude to Professor Chin-Rong Lee for his invaluable suggestions on the manuscript and for several enlightening conversations that significantly enhanced this work. The author is also deeply grateful to the anonymous referees for their numerous insightful comments and constructive suggestions, which have greatly strengthened the content of this paper.

\nocite{*}

\bibliography{refs}{}

\end{document}